\documentclass[journal,draftcsl,12pt, onecolumn]{IEEEtran}
\usepackage{pifont}
\usepackage{times,amsmath,color,geometry,
amssymb,graphicx,epsfig,cite,psfrag,subfigure,algorithm,multirow}
\newtheorem{theorem}{\underline{Theorem}}
\newtheorem{lemma}{\underline{Lemma}}

\newtheorem{remark}{\underline{Remark}}

\begin{document}

\title{Sensing and Recognition When Primary User Has Multiple Power Levels}
\author{Jiachen Li, Feifei Gao, Tao Jiang, and Wen Chen
\thanks{J. Li and F. Gao are with the Tsinghua National Laboratory for Information Science and Technology, Department of Automation, Tsinghua University, Beijing, 10084, China (e-mail: lijc10@mails.tsinghua.edu.cn,
feifeigao@ieee.org).}
\thanks{T. Jiang is with the Wuhan National Laboratory  for
Optoelectronics, Department of Electronic and Information, Huazhong
University of Science and Technology, Wuhan, 430074, China (e-mail:
Tao.Jiang@ieee.org). }
\thanks{W. Chen is with the Department of Electronic Engineering,
Shanghai Jiao Tong University, Shanghai 200240, China (e-mail:
wenchen@sjtu.edu.cn). } } \maketitle \thispagestyle{empty}
\vspace{-10mm}

\begin{abstract}
In this paper, we present a new cognitive radio (CR) scenario when
the primary user (PU) operates under more than one transmit power
levels. Different from the existing studies where PU is assumed to
have only one constant transmit power, the new consideration well
matches  the practical standards, i.e., IEEE 802.11 Series, GSM,
LTE, LTE-A, etc., as well as the adaptive power concept that has
been studied  over the past decades.  The primary target in this new
CR scenario is, of course, still to detect the presence of PU.
However, there appears a secondary target as to identify the PU's
transmit power level. Compared to the existing works where the
secondary user (SU) only senses the ``on-off'' status of PU,
recognizing the power level of PU achieves more ``cognition", and
could be utilized to protect different powered PU with different
interference levels. We derived quite many closed-form results for
either the threshold expressions or the performance analysis, from
which many interesting points and discussions are raised. We then
further study the cooperative sensing strategy in this new cognitive
scenario and show its significant difference from traditional
algorithms. Numerical examples are provided to corroborate the
proposed
studies. 
\end{abstract}

\begin{keywords}
Spectrum sensing, multiple primary transmit power (MPTP), multiple
thresholds, multiple hypothesis testing, power mask, cooperative
sensing, majority decision.
\end{keywords}

\section{Introduction}
Cognitive radio (CR) has been recognized as a promising solution to
spectrum scarcity and spectrum under-utilization \cite{S.Haykin} by
allowing the secondary users (SUs) to access the spectrum of the
primary user (PU) when the latter is  idle. The key component of CR
is thus the  spectrum sensing that could detect whether a specific
frequency band is being used by PU or not.

Popular spectrum sensing techniques that have been proposed include
matched filter detection \cite{MatchedFilter}, energy detection
\cite{energydetection1,energydetection2,energydetection3}, and
cyclostationary detection \cite{Cabric,cyclo2}, among which energy
detection has received intensive attention because it requires the
least prior knowledge of PU and is very simple to implement. The
only prerequisite of the energy detection is that SU knows noiseless
receive power (possibly need to know the transmit power of PU and
path statistics from PU to SU or their product), and then one can
derive a threshold $\lambda$ such that the truly received energy
being greater than $\lambda$ tells the \emph{presence} of PU or
otherwise tells the \emph{absence} of PU. When SU is equipped with
multiple antennas, a promising sensing  technology was designed in
\cite{Eigen} where the eigenvalues of the receive covariance matrix
are used to judge the status of PU.  It is shown that the
performance of multiple antenna based spectrum sensing is much
better than that from a single antenna \cite{multiantenna} because
the former fully utilizes the correlation among antennas. On the
other side, when multiple antenna is infeasible due to the size
limitation of the wireless terminal, one can refer to the
cooperative spectrum sensing from more than one SUs to enhance the
sensing performance \cite{ZhangWei,cooperative2,cooperative3}.

There are also quite a number of works related to CR over the past
ten years. For example, parameter uncertainty based spectrum sensing
\cite{uncertainty}, design with imperfect sensing
\cite{imperfect}, sensing throughput tradeoff \cite{tradeoff},
sensing based sharing \cite{sensingsharing}, spectrum sharing
\cite{spectrumsharing}, as well as many other hybrid schemes between
CR and other technologies, e.g., games among SUs \cite{game},
sensing in OFDM system \cite{OFDM}, sensing in relay network
\cite{relay}, and many others. We may confidently claim that CR has
opened a new research field in wireless communications and has
achieved fruitful research results.

However, it is not difficult to notice that all the existing
spectrum sensing techniques
\cite{MatchedFilter}--\!\!\cite{cooperative3} as well as the related
studies \cite{uncertainty}--\!\!\cite{relay} assume PU either be
absent or be present with a constant power. Yet, it can be easily
known from the current standards, i.e., IEEE 802.11 series
\cite{IEEE}, GSM \cite{GSM}, and the future standards, i.e.,  LTE
\cite{LTE}, LTE-A \cite{LTEA} that the licensed users could be
working under different transmit power levels in order to cope with
different situations, e.g.,
environment, rate, etc. 
A typical example is in CDMA \cite{CDMA} uplink scenario when the
users are subjected to the power control to cope with the near-far
effects. In fact, varying transmit power of the licensed user is a
natural functionality and should be taken into consideration as
witnessed by so many existing literatures studying the power
allocation problem \cite{power1,power2}. Therefore, the traditional
spectrum sensing techniques, which only considers a constant power
level of PU, is not adequate to match both the practical situations
and the theoretical demands.

In fact, FCC \cite{TVwhite1} has specified the interference
protection requirements for TV white space, and the requirements for
different powered services, e.g., full-power digital TV, full-power
analog TV, low-power analog TV, low-power digital TV, etc., are
different as shown in Table \ref{Tab:1}. It is then clear that by
detecting the power levels of PU, SU could adjust its transmit power
to meet the interference requirement for different powered PU.  
Therefore, the sensing target when PU has multiple power levels
should not only be detecting the on-off status of PU but also to
identify its power level.  To the best of the authors' knowledge,
the only work that considers this multiple power transmit power
(MPTP) scenario is \cite{Chenzhong}, where the authors briefly
present  the new sensing strategy while focus more on the optimal
power allocation of SU after power-level recognition  in order to
maximize secondary throughput. Unfortunately, \cite{Chenzhong} does
not fully discuss the fundamentals issues of spectrum sensing in
MPTP and possesses many careless results.

In this paper, we provide  a thorough investigation  over the
spectrum sensing in MPTP scenario and design  two different sensing
strategies. We derive closed form expressions of decision regions
and discussed the  power-mask effect, which is shown to be an unique
phenomenon in MPTP scenario. We also provide many remarks explaining
the fundamental reasoning behind the multiple hypothesis detection
in MPTP based CR. To improve the sensing performance, we further
propose two cooperative sensing schemes, which show much difference
from the traditional cooperative sensing. The closed-form
performance analysis of all four sensing algorithms are derived too.
Various numerical examples are provided to corroborate the proposed
studies.


The rest of the paper is organized as follows. Section II presents
the system model of  MPTP scenario. In Section III, we propose two
different spectrum sensing  strategies and discuss their
relationship. In Section IV, we investigate  cooperative sensing in
MPTP scenario and derive two different algorithms based on the
majority voting and the MAP detection, respectively. In Section V,
simulations results are provided to evaluate the designed
algorithms. Finally, conclusions are drawn in Section VI.


\section{System Model}
Consider a simple CR network that consists of one PU, $K$ SUs and a
common receiver, i.e., the decision fusion center. The primary user
(PU) could either be absent or operate under one power-level $P_i,
i\in\{1,2,\ldots,N\}$. Without loss of generality, we assume
$P_{i+1}>P_{i}>0,\forall i$. As one of the earliest work considering
the MPTP scenario, we assume the values of the noiseless received
powers $\gamma_kP_i, \ \forall i$, are known at SU-$k$.  More
considerable situations, e.g.,   partially known channels or power
levels could serve as future research topics.

During the sensing time, the
$l$th received sample at SU-$k$ can be expressed as:
\begin{align}\label{eq:system model}
x_{l,k}=\left\{ \begin{array}{ll} n_{l,k}& {\cal H}_{0} \\
\sqrt{P_{i}}\sqrt{\gamma_k}s_{l,k}+n_{l,k}&{\cal H}_{i}, \ i =
1,2,...,N\end{array}\right.
\end{align}
where ${\cal H}_{0}$ denotes  the hypothesis that PU is absent while
${\cal H}_i$ indicates  PU is operating under   power-level $P_i$;
$s_{l,k}$ is the $l$th sample transmitted from PU, which is assumed
to follow complex Gaussian distribution with zero mean and unit
variance, i.e., $s_{l,k}\sim {\cal CN}(0,1)$; $n_l$ is the additive
noise that follows ${\cal CN}(0,\sigma^2_n)$ for all cases. If we
define $P_0=0$ as the power when PU is absence, then a unified
expression can be obtained as
\begin{equation}\label{eq:sample distribution}
x_{l,k}\sim{\cal CN}(0,\gamma_k P_i+\sigma^2_n), \quad \forall {\cal
H}_i,\ i\in\{0,1,\ldots,N\},
\end{equation}

Let us  define the prior probability of each state of PU as
$\text{Pr}({\cal H}_i),\ i=0,1,\ldots,N$. Then, the presence state
of PU will include all ${\cal H}_i, i\geq 1$ and will be denoted as
${\cal H}_\text{on}$. Obviously, ${\cal H}_\text{on}$  has the prior
probability $\text{Pr}({\cal H}_\text{on})=\sum_{i=1}^N
\text{Pr}({\cal H}_i)$, while the absence state of PU, denoted by
${\cal H}_\text{off}\triangleq{\cal H}_0$,  has the probability
$\text{Pr}({\cal H}_\text{off})=\text{Pr}({\cal H}_0)$.

In MPTP scenario, we define the primary target of spectrum sensing
as detecting the presence of PU, while define the secondary target
as recognizing the power-level of PU.  As mentioned in Section I, PU
may operate in different power levels, each with a different
tolerable interference level from SU. After recognizing the power
level of PU, SU could choose a proper transmit power to fulfill the
interference requirements.

\section{Spectrum Sensing at Local SU}
Let us first present how local secondary user (e.g., SU-$k$)
performs spectrum sensing in MPTP scenario. The user index $k$ is thus dropped for
notation conciseness. We propose the following two different but relevant
approaches.

\subsection{Sensing Strategy-I:  Detecting the Presence First }

Since we define the primary task in MPTP as to check the presence of
PU, we may  first verify the hypothesis  ${\cal H}_{\text{on}}/{\cal
H}_{\text{off}}$. If ${\cal H}_{\text{on}}$ is detected, then the
next step is to recognize which ${\cal H}_i, i\geq 1$ is true.
Let us assume that SU receives a total number of $M$ samples during
the sensing period, denoted as  $\mathbf{x}=[x_1,x_2,\ldots,x_M]^T$.
The ratio of the posterior probabilities between two hypothesis can
be written as
\begin{align}\label{eq:likelihood}
\eta(\mathbf{x})=&\frac{\text{Pr}({\cal
H}_\text{on}|\mathbf{x})}{\text{Pr}({\cal
H}_\text{off}|\mathbf{x})}=\frac{\sum_{i=1}^N \text{Pr}({\cal
H}_i|\mathbf{x})}{\text{Pr}({\cal
H}_0|\mathbf{x})}=\sum_{i=1}^N\frac{ p(\mathbf{x}|{\cal
H}_i)\text{Pr}({\cal H}_i)}{p(\mathbf{x}|{\cal H}_0)\text{Pr}({\cal
H}_0)}
\nonumber\\
=&\sum_{i=1}^N\frac{\text{Pr}({\cal H}_i)}{\text{Pr}({\cal
H}_\text{off})}\bigg(\frac{\sigma^2_n}{\gamma P_i
+\sigma^2_n}\bigg)^M\text{exp}\bigg\{\frac{\gamma P_i\sum_{l=1}^M
|x_l|^2}{\sigma^2_n(\gamma P_i+\sigma^2_n)} \bigg\}.
\end{align}
It is easily seen that $\eta(\mathbf{x})$ is strictly increasing
over $y\triangleq\sum_{l=1}^M |x_l|^2$, i.e., the received energy,
so the decision can be alternatively made through
\begin{equation}
y\mathop{\gtrless}\limits_{\mathcal{H}_\text{off}}^{\mathcal{H}_\text{on}}\theta,
\label{eq:y}
\end{equation}
where $\theta$ is the pre-determined parameter. Hence, the optimal
detector is the energy detector and we can re-represent
$\eta(\mathbf{x})$ as $\eta(y)$ . The parameter $\theta$ is used to
control the detection performance. For example if $\theta$ is used
to control the false alarm probability, then the detection follows
Neyman Pearson rule; If $\theta$ is set such that $\eta(y)=1$, then
the detection follows the maximum a posterior (MAP) rule .

The probability density functions (pdf) of $y$ conditioned on
$\mathcal{H}_i,\ i\in\{0,1,\ldots,N\}$ and $\mathcal{H}_\text{on}$
can be derived as
\begin{align}\label{eq:pdf Hn}
&p(y|{\cal H}_i)=\frac{y^{M-1}e^{-\frac{y}{\gamma
P_i+\sigma^2_n}}}{\Gamma(M)(\gamma
P_i+\sigma^2_n)^{M}}, \\
&\label{eq:pdf Hon} p(y|{\cal H}_\text{on})=\sum_{i=1}^{N} p(y|{\cal
H}_i)\frac{\text{Pr}{({\cal H}_i)}}{\text{Pr}({\cal
H}_\text{on})}=\frac{1}{\text{Pr}({\cal H}_\text{on})}\sum_{i=1}^{N}
p(y|{\cal H}_i)\text{Pr}{({\cal H}_i)},
\end{align}
where $\Gamma(\cdot)$ denotes the gamma function.


Similar to the conventional CR, we could resort to the false alarm
probability and the detection probability to describe the
performance of the detection,  separately calculated as
\begin{align}
\text{P}_\text{fa}(\theta)&=\text{Pr}(\mathcal{H}_\text{on}|\mathcal{H}_\text{off})=
\int_\theta^\infty
p(y|\mathcal{H}_\text{off})dy=\frac{\gamma(M,\frac{\theta}{\sigma^2_n})}
{\Gamma(M)} ,\label{eq:false alarm}\\
\text{P}_\text{d}(\theta)&=\text{Pr}(\mathcal{H}_\text{on}|\mathcal{H}_\text{on})=
\frac{1}{\text{Pr}(\mathcal{H}_\text{on})}\sum_{i=1}^N
\int_\theta^\infty\!\!
p(y|\mathcal{H}_i)\text{Pr}(\mathcal{H}_i)dy=\sum_{i=1}^N\!\!\frac{\gamma(M,\frac{\theta}{\gamma
P_i+\sigma^2_n})}{\Gamma(M)}
\frac{\text{Pr}(\mathcal{H}_i)}{\text{Pr}(\mathcal{H}_\text{on})}\label{eq:detectionprobabilty},
\end{align}
where $\gamma(\cdot,\cdot)$ denotes the lower incomplete gamma
function and $\theta$. As usual, one can adjust $\theta$ to achieve
a desired $\text{P}_\text{fa}(\theta)$,  and then the  corresponding
$\text{P}_\text{d}(\theta)$ is immediately settled.


Let us compute the threshold $\theta_\text{on/off}$ according to MAP criterion for  consistency  in the rest of the paper. From
(\ref{eq:likelihood}),  $\theta_\text{on/off}$ could be obtained from  $\eta(\theta_\text{on/off})=1$ and can be numerically
computed from the equation $\Phi(\theta)=0$, where $\Phi(\theta)$ is
defined as
\begin{equation}\label{def:Phi}
\Phi(\theta)\triangleq\sum_{i=1}^N \frac{\text{Pr}({\cal
H}_i)}{(\frac{\gamma_k P_i}{\sigma^2_n} +1)^M} \text{\large
e}^{\frac{\gamma_k P_i}{\sigma^2_n(\gamma_k
P_i+\sigma^2_n)}\cdot\theta}-\text{Pr}({\cal H}_0).
\end{equation}
It can be easily checked that
$\frac{\partial\Phi(\theta)}{\partial\theta}>0$  and $\Phi(0)<0$ as
long as $M$ is sufficiently large. Hence the solution
$\theta_\text{on/off}$ that makes $\Phi(\theta)=0$ must exist and is
definitely unique as well.

If the received energy satisfy $y>\theta_\text{on/off}$, then PU is
claimed to be present and the next step is to recognize which
power-level of PU is in use. A natural approach is to formulate
multiple hypothesis testing \cite{Kay} and apply MAP detection,
where for a hypothesis pair $({\cal H}_i, {\cal H}_j)$, $\forall
i,j\geq 1 $, ${\cal H}_i$ beats  ${\cal H}_j$  if
\begin{equation}\label{MAP multiple}
\text{Pr}({\cal H}_i|\mathbf{x},\hat{{\cal
H}}_\text{on})>\text{Pr}({\cal H}_j|\mathbf{x},\hat{{\cal
H}}_\text{on}).
\end{equation}
Here,  we use $\hat{{\cal H}}_\text{on}$ to denote that the presence
detection has been made already.\footnote{Please not that
$\hat{{\cal H}}_\text{on}$ is not the same as ${\cal H}_\text{on}$.}
From Bayes rule, there is
\begin{align}
\text{Pr}({\cal H}_i|\mathbf{x},\hat{{\cal
H}}_\text{on})=&\frac{p(\mathbf{x}|{\cal H}_i,\hat{{\cal
H}}_\text{on})\text{Pr}({\cal H}_i|\hat{{\cal
H}}_\text{on})}{p(\mathbf{x}|\hat{{\cal H}}_\text{on})} \label{bayes
rule 3}.
\end{align}
Let us define the equivalent region of $\mathbf{x}\in{\cal X}$ to
$y>\theta_{\text{on/off}}$, and then represent $\hat{{\cal
H}}_\text{on}$ as $\mathbf{x}\in{\cal X}$. Then (\ref{bayes rule 3})
can be rewritten as
\begin{align}
\text{Pr}({\cal H}_i|\mathbf{x},\hat{{\cal
H}}_\text{on})=&\frac{1}{p(\mathbf{x}|\mathbf{x}\in{\cal X})}\cdot
p(\mathbf{x}|{\cal H}_i,\mathbf{x}\in{\cal X})\cdot\text{Pr}({\cal
H}_i|\mathbf{x}\in{\cal X})\nonumber\\
=&\frac{1}{p(\mathbf{x}|\mathbf{x}\in{\cal
X})}\cdot\frac{p(\mathbf{x}|{\cal
H}_i)}{\text{Pr}(\mathbf{x}\in{\cal X}|{\cal
H}_i)}\cdot\frac{\text{Pr}(\mathbf{x}\in{\cal X}|{\cal
H}_i)\text{Pr}({\cal H}_i)}{\text{Pr}(\mathbf{x}\in{\cal
X})}\nonumber\\
=&\frac{p(\mathbf{x}|{\cal H}_i)\text{Pr}({\cal
H}_i)}{p(\mathbf{x}|\mathbf{x}\in{\cal
X})\text{Pr}(\mathbf{x}\in{\cal X})}.\label{eq:i1}
\end{align}
Note that the following equality holds from the definition of
probability density function
\begin{equation}
p(\mathbf{x}|{\cal H}_i,\mathbf{x}\in{\cal
X})=\frac{p(\mathbf{x}|{\cal H}_i)}{\text{Pr}(\mathbf{x}\in{\cal
X}|{\cal H}_i)}, \quad \mathbf{x}\in{\cal X} \label{eq:i2}
\end{equation}
and is used to derive (\ref{eq:i1}). We place $\mathbf{x}\in{\cal
X}$ in (\ref{eq:i2}) to represents that (\ref{eq:i2}) holds only for
domain $\mathbf{x}\in{\cal X}$.


Therefore, the MAP detection  (\ref{MAP multiple}) is simplified to
\begin{align}
 p(\mathbf{x}|{\cal H}_i)\text{Pr}({\cal H}_i)>p(\mathbf{x}|{\cal
H}_j)\text{Pr}({\cal H}_j),\quad \mathbf{x}\in{\cal X}  \label{MAP
multiple_simplified}
\end{align}
\begin{remark}
From (\ref{MAP multiple_simplified}), we know the MAP detection for
power levels is not related with how ${\cal H}_{\text{on}}$ is
detected, i.e., we can apply either MAP detection or Neyman Pearon
to check the presence of PU without affecting (\ref{MAP
multiple_simplified}). Nevertheless, the way to detect the presence
of PU will affect the value of $\text{Pr}({\cal
H}_i|\mathbf{x},\hat{{\cal H}}_\text{on})$ as seen from
(\ref{eq:i1}).
\end{remark}

Hence, the MAP detection of the power level can be simply  described
as
\begin{align}
\hat{i}=\arg\max_{i\in\{1,\ldots, N\}}\  p(\mathbf{x}|{\cal
H}_i)\text{Pr}({\cal H}_i), \quad \mathbf{x}\in{\cal
X}.\label{eq:MAP-final}
\end{align}


Let us then define the ratio
\begin{align}\label{eq:likelihood2}
&\xi(\mathbf{x})=\frac{ p(\mathbf{x}|{\cal H}_i)\text{Pr}({\cal
H}_i)}{ p(\mathbf{x}|{\cal H}_j)\text{Pr}({\cal
H}_j)}=\frac{\text{Pr}(\mathcal{H}_i)}{\text{Pr}(\mathcal{H}_j)}\bigg(\frac{\gamma
P_j+\sigma^2_n}{\gamma P_i
+\sigma^2_n}\bigg)^M\text{exp}\bigg\{\frac{\gamma(P_i-P_j)
\sum_{l=1}^M |x_l|^2}{(\gamma P_i+\sigma^2_n)(\gamma
P_j+\sigma^2_n)} \bigg\}.
\end{align}
Obviously, $\xi(\mathbf{x})$ is purely related with the energy
$y=\sum_{l=1}^M |x_l|^2$ (other variables are constants). Hence, the
energy detector is again optimal when recognizing the power level of
PU, and we can represent $\xi(\mathbf{x})$ by $\xi(y)$. Since
$\xi(\mathbf{x})$ is an increasing function of $y$ when $P_i>P_j$,
one can easily know that the decision region of ${\cal H}_i$,
denoted as ${\cal R}(\mathcal{H}_i)$, must be a continuous region of
$y$, and ${\cal R}(\mathcal{H}_i)$ must stay on the right side of
${\cal R}(\mathcal{H}_j)$ if $P_{i}>P_{j}$.
\begin{theorem}\label{theorem:decision range}
The decision regions  of hypothesis $\mathcal{H}_i,\
i\in\{1,2,\ldots,N\}$ are
\begin{align}\label{decision range}
{\cal R}(\mathcal{H}_i):=\left\{\begin{array}{ll}
y\in\big(\theta_{\text{on/off}},\ \min_{1<j\leq
N}\Theta(1,j)\big)&i=1\\
y\in\big(\max\{\theta_{\text{on/off}},\max_{1\leq
j<i}\Theta(i,j)\},\ \min_{i<j\leq
N}\Theta(i,j)\big),&1<i<N\\
y\in\big(\max\{\theta_{\text{on/off}},\max_{1\leq
j<N}\Theta(N,j)\},\ +\infty\Big),&i=N\end{array}\right.
\end{align}
where $\Theta(i,j)$ is defined as
\begin{align}\label{df:big theta}
&\Theta(i,j)\triangleq\frac{(\gamma P_i+\sigma^2_n)(\gamma
P_j+\sigma^2_n)}{\gamma(P_i-P_j)}\ln\bigg[\bigg(\frac{\gamma
P_i+\sigma^2_n}{\gamma
P_j+\sigma^2_n}\bigg)^{M}\frac{\text{Pr}({\cal
H}_j)}{\text{Pr}({\cal H}_i)}\bigg].
\end{align}
\end{theorem}

\begin{proof}
Substituting (\ref{eq:likelihood2}) into (\ref{MAP multiple_simplified})
yields
\begin{equation*}
\frac{y\cdot\gamma(P_i-P_j)}{(\gamma P_i+\sigma^2_n)(\gamma
P_j+\sigma^2_n)}>\ln\bigg[\bigg(\frac{\gamma P_i+\sigma^2_n}{\gamma
P_j+\sigma^2_n}\bigg)^{M}\frac{\text{Pr}({\cal
H}_j)}{\text{Pr}({\cal H}_i)}\bigg], \quad \forall i,j\geq 1.
\end{equation*}
If $P_i>P_j$, i.e., $i>j$, then there is
\begin{equation*}
y>\frac{(\gamma P_i+\sigma^2_n)(\gamma
P_j+\sigma^2_n)}{\gamma(P_i-P_j)}\ln\bigg[\bigg(\frac{\gamma
P_i+\sigma^2_n}{\gamma
P_j+\sigma^2_n}\bigg)^{M}\frac{\text{Pr}({\cal
H}_j)}{\text{Pr}({\cal H}_i)}\bigg], \quad \forall i>j;
\end{equation*}
If $P_i<P_j$, i.e., $i<j$,  then  there is
\begin{equation*}
y<\frac{(\gamma P_i+\sigma^2_n)(\gamma
P_j+\sigma^2_n)}{\gamma(P_i-P_j)}\ln\bigg[\bigg(\frac{\gamma
P_i+\sigma^2_n}{\gamma
P_j+\sigma^2_n}\bigg)^{M}\frac{\text{Pr}({\cal
H}_j)}{\text{Pr}({\cal H}_i)}\bigg] , \quad \forall i<j.
\end{equation*}

Then for $1<i<N$, the lower bound of ${\cal R}(\mathcal{H}_i)$
should be $y>\max\limits_{1\leq j<i}\Theta(i,j)$ and the upper bound
should be $y<\min\limits_{i<j\leq N}\Theta(i,j)$. Moreover, the MAP
detection is defined on the domain $\mathbf{x}\in{\cal X}$, i.e.,
$y>\theta_{\text{on/off}}$, so all  decision regions of non-zero
power should stay in $(\theta_{\text{on/off}},+\infty)$. Bearing in
mind that $\theta_{\text{on/off}}$ may be greater than
$\max\limits_{1\leq j<i}\Theta(i,j)$ for some $i$, the proof is
completed.
\end{proof}
\begin{remark}
The decision region of ${\cal H}_0$, i.e., the absence of PU can be
expressed in a unified way as
\begin{align}
{\cal R}(\mathcal{H}_0):=  y\in\big(0,\ \theta_{\text{on/off}}\big).
\end{align}

\end{remark}

\begin{remark}
Compared to the traditional ``on-off'' based sensing  that has only
one threshold,  the new scenario MPTP needs multiple thresholds to
separate different power levels, as shown in Fig.
\ref{fig:thresholds}, where $\theta_1\triangleq
\theta_{\text{on/off}}$ is defined for consistence.
\end{remark}

An interesting and special phenomenon in MPTP happens when the
computed lower bound of a specific region ${\cal R}(\mathcal{H}_{i_0})$
is greater than the upper bound, e.g.,
\begin{equation}\label{submerge criteria} \max\{\theta_{\text{on/off}},\ \max_{
1\leq j<i_0}\Theta(i_0,j)\}>\min_{i_0<j\leq N}\Theta(i_0,j)
\end{equation}
holds for some specific $1<i_0<N$. Once this happens then ${\cal R}(\mathcal{H}_{i_0})$
is empty and the power level
$P_{i_0}$ can never be detected. We call this new phenomenon in
MPTP as \emph{power-mask} effect.  Hence, in   Fig.
\ref{fig:thresholds}., the number of the thresholds may be less than or equal to $N$.

\begin{remark}
If $\theta_{\text{on/off}}>\max_{ 1\leq
j<i_0}\Theta(i_0,j)>\min_{i_0<j\leq N}\Theta(i_0,j)$, then $P_{i_0}$
is masked from left by $P_0$, while if $\max_{ 1\leq
j<i_0}\Theta(i_0,j)>\theta_{\text{on/off}}>\min_{i_0<j\leq
N}\Theta(i_0,j)$, then $P_{i_0}$ is masked from both sides by $P_{i_{0}-1}$ and
$P_{i_{0}+1}$.
\end{remark}

Some intuitive explanation for power mask is provided here. First
note that the bounds in the decision region ${\cal
R}(\mathcal{H}_i)$ are affected by many parameters, i.e., $P_i$,
$\gamma$, $\sigma_n^2$, $\text{Pr}({\cal H}_i)$. If the prior
probability $\text{Pr}({\cal H}_{i_0})$ is very small, i.e., the
power level $P_{i_0}$ is seldom used by PU. Then $P_{i_0}$ may
easily be ``ignored'' by SU and is then masked. Another example is
that, if $P_{i_0}$ is closed to $P_{i_0-1}$ and $P_{i_0+1}$ and if
$\sigma_n^2$ is larger, then it is very possible that $P_{i_0}$ will
be masked by $P_{i_0-1}$ or $P_{i_0+1}$ due to the large uncertainty
caused by the noise.

\begin{remark}
Note that the leftmost level $P_0$ and the rightmost level $P_N$ cannot be masked and are always
detectable.
\end{remark}

\begin{remark}
When power-mask happens for a specific $i_0$, it means that the
$P({\cal H}_{i_0}|\mathbf{x})$ cannot beat any other $P({\cal
H}_{i}|\mathbf{x}),i\neq i_0$ but $P({\cal
H}_{i_0}|\mathbf{x})$ could still  possess non-zero value. Hence, it is
possible to design some sophisticated approach which considers this ``soft''
information and remove the power-mask effect. Nevertheless, the
corresponding discussion is out of the scope of this paper and will
be left for future research.
\end{remark}
\begin{remark}
It is of interest to check whether the decision regions for ${\cal
H}_i$'s are continuously connected for two consecutively detectable
indices $1<i_0<i_0+1<N$, i.e., for those not masked power level, check whether $ \min_{i_0<j\leq N}\Theta(i_0,j)=\max_{
1\leq j<i_0+1}\Theta(i_0+1,j)$ holds. Unfortunately, due to the discrete nature of
the power mask effect, we cannot mathematically prove this property.
Nevertheless, it can be easily known that, for any $y$, there is
always a corresponding decision according to the MAP detection
(\ref{eq:MAP-final}). Therefore, there should be no gap between any
two consecutive decision regions.
\end{remark}

A special case that affects the power-mask effect appears when
$\text{Pr}(P_i)=\text{Pr}(P_j),\ \forall i,j\in\{1,2,\ldots,N\}$.
\begin{lemma}
If $\text{Pr}(P_i)=\text{Pr}(P_j)$ holds for $\forall
i,j\in\{1,2,\ldots,N\}$, then $\Theta(i,j)$ is not related to
$P({\cal H}_i)$ and is an increasing function over $P_j$ for any
$P_i$.
\end{lemma}
\begin{proof}
In this case, take the partial derivative of $\Theta(i,j)$ over
$P_j$ as
\begin{align}
\frac{\partial\Theta(i,j)}{\partial P_j}&=\frac{M(\gamma
P_i+\sigma^2_n)^2}{\gamma(P_j-P_i)^2}\bigg\{\frac{\gamma(P_j-P_i)}{\gamma
P_i+\sigma^2_n}-\ln\Big[1+\frac{\gamma(P_j-P_i)}{\gamma
P_i+\sigma^2_n}\Big]\bigg\}.
\end{align}
Redefining $z=\frac{\gamma(P_j-P_i)}{\gamma P_i+\sigma^2_n}$, it can
be easily known  that $z\in(-1,+\infty)$ for any $i$. We then obtain
$$\frac{\partial\Theta(i,j)}{\partial P_j}=\frac{M\gamma}{z^2}\big[z-\ln(1+z)\big].$$
It is also clear that $z-\ln(1+z)\geq0$ for $z>-1$. Hence,
$\frac{\partial\Theta(i,j)}{\partial P_j}\geq 0$ holds\footnote{Note that for $z=0$ the value of $\frac{\partial\Theta(i,j)}{\partial P_j}$ is obtained from Hospital's rule and is $\frac{M\gamma}{2}>0$.}   for all
possible $P_j$ and $\Theta(i,j)$ is an increasing function over
$P_j$.
\end{proof}

According to Lemma 1,  when $\text{Pr}(P_i)=\text{Pr}(P_j),i,j\geq 1$ there is
\begin{align}
\label{eq:special} \max_{1\leq
j<i}\Theta(i,j)=\Theta(i,i-1)<\Theta(i,i+1)=\min_{i<j\leq
N}\Theta(i,j).
\end{align}
Therefore, the non-zero power levels cannot mask each other, while
the power mask effect may only happen when $P_0$ mask the power
levels on its right side.

\begin{remark}
When $\text{Pr}(P_i)=\text{Pr}(P_j),i,j\geq 1$, denote the first
power level that is not masked by $P_0$ as $P_{i_0}$. It can be
readily check that  all the power levels on its right sides, i.e.,
$P_i,i>i_0$ has $\Theta(i,i+1)=\Theta(i+1,i)$, and hence the
decision regions for all ${\cal H}_i$'s are mathematically proved to
be connected in this case.
\end{remark}

Note that the decision range ${\cal R}(\mathcal{H}_i)$ can be
determined in a prior manner by calculating all $\Theta(i,j)$ and
$\theta_{\text{on/off}}$ in advance.  To unify  our discussion, let
us use $\theta_i, i\in\{1,2,\ldots,N\}$ to denote the threshold
between ${\cal R}(\mathcal{H}_{i-1})$ and ${\cal R}(\mathcal{H}_i)$, and define
$\theta_0\triangleq0$, $\theta_{N+1}\triangleq+\infty$ for
completeness. Meanwhile, if $P_{i_0}$ cannot be detected due to
power-mask effect, we denote $\theta_{i_0}=\theta_{i0+1}$ so that
the corresponding decision range $[\theta_{i_0}, \theta_{i_0+1})$ is
empty.  

To characterize the performance of the  spectrum sensing in MPTP, purely resorting to
$\text{P}_\text{d}(\theta)$ and $\text{P}_\text{fa}(\theta)$ is not adequate.  We
should calculate all the probabilities when SU makes the decision as
hypothesis $\mathcal{H}_j$ while PU is actually transmitting with
$P_i$, that is
\begin{align}\label{decision probability}
\text{Pr}(\mathcal{H}_j|\mathcal{H}_i)=&\int_{\theta_j}^{\theta_{j+1}}p(y|\mathcal{H}_i)dy=\frac{\gamma\left(M,\frac{\theta_{j}}{\gamma
P_i+\sigma^2_n}\right)} {\Gamma(M)}-\frac{\gamma\left(M,\frac{\theta_{j+1}}
{\gamma P_i+\sigma^2_n}\right)}{\Gamma(M)},\quad \forall i,j,
\end{align}
which is also defined  as \emph{decision probability} in this paper.
Obviously, the decision probabilities of those masked power
are zeros.


Then, $\text{P}_\text{d}(\theta)$ and $\text{P}_\text{fa}(\theta)$
can be easily obtained from the summations of the corresponding
$\text{Pr}(\mathcal{H}_j|\mathcal{H}_i)$. Moreover we may  introduce
a new technical term
\begin{equation}
P_\text{dis1}
=\frac{1}{\text{Pr}(\cal
H_\text{on})}\sum_{i=1}^{N}\text{Pr}({\cal H}_i|{\cal H}_i)P({\cal
H}_i),\label{eq:P_dis1}
\end{equation}
named as discrimination probability to describe
the recognition capability for our secondary target.

\subsection{Sensing Strategy II: Recognize Power Level First}

Another reasonable approach  to achieve both our targets is to directly detecting the
power level of PU by treating $P_0$ as an \emph{equivalent} power level (but with zero value) as
other non-zero $P_i, i\geq 1$. The presence or the absence
can be immediately found after the power level index is detected.

%

From MAP based multiple hypothesis testing, the optimal detection can be
stated as
\begin{align}\label{MAP-2}
\hat{i}=\arg\max_{i\in\{0,1,\ldots,N\}}\ \text{Pr}({\cal H}_i|\mathbf{x}).
\end{align}

Since the expression of (\ref{MAP-2}) is, mathematically, the
same as (\ref{eq:MAP-final}) but includes one more index $0$, the previous results can be immediately
modified  here. For example, the decision region is computed as
\begin{align}\label{decision range2}
{\cal R}(\mathcal{H}_i):=\left\{\begin{array}{ll}
y\in\big(0,\ \min_{0<j\leq
N}\Theta(0,j)\big)&i=0\\
y\in\big(\max_{0\leq j<i}\Theta(i,j),\ \min_{i<j\leq
N}\Theta(i,j)\big),&0<i<N\\
y\in\big(\max_{0\leq j<N}\Theta(N,j),\
+\infty\Big),&i=N\end{array}\right.
\end{align}
where $\Theta(i,j)$ is given in (\ref{df:big theta}). Hence, the
power-mask effect also exists if $\max_{ 1\leq j<i_0}\Theta(i_0,j)
>\min_{i_0<j\leq N}\Theta(i_0,j)$ for some $i_0$. Let us use
$\phi_i$ to represent  the thresholds separating ${\cal
R}(\mathcal{H}_{i-1})$ and ${\cal R}(\mathcal{H}_i)$. Then the
decision probability is given in (\ref{decision probability}) with
$\theta_i$ being replaced by $\phi_i$.

\begin{remark}
In sensing strategy-II, once again, both  $P_0$ and $P_N$ cannot be
masked and are always detectable.
\end{remark}

Due to the similarity between the decision regions of the
two sensing approaches, i.e., (\ref{decision range}) and
(\ref{decision range2}), one natural and interest  question  arises:
are all the thresholds or parts of the thresholds the same in these two sensing approaches?
Of all the thresholds, the first one
$\theta_\text{on/off}$ that separates the absence decision and
presence decision is of special importance. Then we  provide the
following lemma.
\begin{lemma}\label{lemma:2}
Let $\phi_\text{on/off}\triangleq\phi_1$ be the threshold between
detecting absence and detecting   presence in sensing strategy II,
there is $\phi_\text{on/off}>\theta_\text{on/off}$.
\end{lemma}
\begin{proof}
Considering the power-mask effect,  $\phi_\text{on/off}$ can  be computed from $\text{Pr}({\cal H}_0|\mathbf{x})=\text{Pr}({\cal H}_{j_0}|\mathbf{x})$ where $j_0$ is the index
of the first non-zero power that is not masked. From (\ref{eq:likelihood2}), we can derive a unique $\phi_\text{on/off}$ as
\begin{equation}\label{eq:phi_on_off}
\phi_\text{on/off}=\frac{\sigma^2_n(\gamma P_{j_0}+\sigma^2_n)}{\gamma
P_{j_0}}\cdot\ln\bigg[\frac{\text{Pr}(\mathcal{H}_0)}{\text{Pr}(\mathcal{H}_{j_0})}\Big(\frac{\gamma
P_{j_0}}{\sigma^2_n}+1\Big)^M\bigg].
\end{equation}
Let us then compute $\Phi(\phi_\text{on/off})$ from (\ref{def:Phi}), which yields
\begin{align*}
\Phi(\phi_\text{on/off})&=\sum_{i=1}^N \frac{\text{Pr}({\cal H}_i)}{(\frac{\gamma P_i}{\sigma^2_n} +1)^M} \text{\large e}^{\frac{\gamma P_i}{\sigma^2_n(\gamma P_i+\sigma^2_n)}\cdot\phi_\text{on/off}}-\text{Pr}({\cal H}_0) \\
&=\sum_{i\neq j_0}^N \frac{\text{Pr}({\cal H}_i)}{(\frac{\gamma
P_i}{\sigma^2_n} +1)^M} \text{\large e}^{\frac{\gamma
P_i}{2\sigma^2_n(\gamma P_i+\sigma^2_n)}\cdot\phi_\text{on/off}}>0.
\end{align*}
Since $\Phi(\theta_\text{on/off})=0$ and $\Phi(\theta)$ is an increasing function over $\theta$, it is obvious that $\phi_\text{on/off}>\theta_\text{on/off}$. The lemma is proved.
\end{proof}

Lemma \ref{lemma:2} suggests that sensing strategy-II claims more absence of PU than sensing strategy-I. 
In general, when we compare (\ref{decision range2}) and
(\ref{decision range}), it is easy to note that   $\max_{0\leq
j<i}\Theta(i,j)$ may be greater than $\max_{1\leq j<i}\Theta(i,j)$
if $\Theta(i,0)$ is the bigger than other $\Theta(i,j),\ 0<j<i$.
When this happens, it also means that $P_0$ masks all power levels
smaller than $P_{i+1}$ in sensing strategy-II. Hence, the lower
bound of ${\cal R}(\mathcal{H}_i)$ in sensing strategy-II is
$\phi_\text{on/off}$ while the lower bound of ${\cal
R}(\mathcal{H}_i)$ in sensing strategy-I is
$\max\{\theta_{\text{on/off}},\ \max_{ 1\leq j<i_0}\Theta(i_0,j)\}$.
Combing Lemma \ref{lemma:2}, it is then clear that the lower bound
of ${\cal R}(\mathcal{H}_i)$ in sensing strategy-II is bigger than
that of sensing strategy-I. Nevertheless, when  $\Theta(i,0)$ is not
the dominant one in $\max_{0\leq j<i}\Theta(i,j)$, i.e.,
$\text{Pr}({\cal H}_0|\mathbf{x})$ is not the biggest among all
$\text{Pr}({\cal H}_i|\mathbf{x})$, then the thresholds for the two
sensing strategies are the same.


\begin{remark}
A special case happens when $\text{Pr}({\cal H}_i)=\text{Pr}({\cal
H}_j), \forall i,j$. In this case, (\ref{eq:special}) holds and
$\Theta(i,j)$ is not related with $\text{Pr}({\cal H}_i)$. Hence,
except $\theta_\text{on/off}$ and $\phi_\text{on/off}$, all the
other thresholds from both sensing strategies are the same, i.e.,
$\theta_i=\phi_i$. Moreover, the power mask effect is completely
removed.
\end{remark}

\begin{remark}
Since we treat $P_0$ as an equal state as other non-zero $P_i$'s in
sensing strategy-II, the discrimination probability in the second
sensing approach could be also be defined as
\begin{equation}\label{eq:P_dis2}
P_\text{dis2}=\sum_{i=0}^{N}\text{Pr}({\cal H}_i|{\cal H}_i)P({\cal
H}_i).
\end{equation}
\end{remark}
%

\subsection{Fundamental Rationale Behind  Two Spectrum Sensing
Strategies} \label{sec:rationale}

After presenting two different spectrum sensing approaches, both
seemingly reasonable, a natural question arises: which one is better
and why? Let us explain from MAP detection point of view.

When MAP detection is applied in the first step of sensing
strategy-I, then the obtained $\theta_{\text{on/off}}$ is optimal in
terms of minimizing the following error
\begin{align}
&\arg\min\ \text{Pr}({\cal H}_{\text{off}}|{\cal
H}_\text{on})\text{Pr}({\cal H}_\text{on})+\text{Pr}({\cal
H}_\text{on}|{\cal H}_\text{off})\text{Pr}({\cal H}_\text{off})\nonumber\\
=&\arg\min\ \sum_{i=1}^N \text{Pr}({\cal H}_0|{\cal
H}_i)\text{Pr}({\cal H}_i)+ \sum_{i=1}^N \text{Pr}({\cal H}_i|{\cal
H}_0)\text{Pr}({\cal H}_0).
\end{align}

On the other side, applying MAP in sensing strategy-II that directly
detects the power level is optimal in terms of minimizing the
following error
\begin{equation}
\arg\min\ \sum_{j\neq i}\text{Pr}({\cal H}_j|{\cal
H}_i)\text{Pr}({\cal H}_i).
\end{equation}

Obviously, sensing strategy-I does not consider the errors
$\text{Pr}({\cal H}_j|{\cal H}_{i}), i,j\geq 1$ when detecting the
presence of PU while   sensing strategy-II takes into account of all
error probabilities $\text{Pr}({\cal H}_j|{\cal H}_{i}),\forall i,
j$ all at once.

In general, if identifying the presence of PU is a more important
target than discriminating the power levels of PU for example in
traditional sensing based CR, then sensing strategy-I is preferable.
However for sensing based sharing CR scheme and when the penalty due
to the wrong interference protection is very high, then the sensing
strategy-II could be preferable.

\begin{remark}
In fact, both the strategies falls into the Baye's Risk based
multiple hypothesis test where the optimization criterion is to
minimize
\begin{equation}
\sum_{i,j}C_{i,j}\text{Pr}({\cal H}_j|{\cal H}_i)\text{Pr}({\cal
H}_i),
\end{equation}
and $C_{i,j}$ is the price or the cost for detecting ${\cal H}_j$
when ${\cal H}_i$ is true. Obviously, the value of $C_{i,j}$ should
be set according to practical requirements and can be different in
different applications.
\end{remark}

\section{Cooperative Spectrum Sensing}
Similar to the conventional cooperative sensing
\cite{ZhangWei,cooperative2,cooperative3}, we assume each SU
performs sensing, either with strategy-I or strategy-II,  and
forwards the result to a fusion center. Note that, local SUs only
need to forward the power-level index $i\in\{0,1,\ldots, N\}$ where
the ``on-off'' information of PU is automatically embedded. It is
then easily known that the existing \textit{Logic-AND} (LA),
\textit{Logic-OR} (LR) and their general form \textit{k out of N}
(KON) based fusion rules, mainly designed for binary results ``0''
and ``1'', are no longer applicable when the forwarded indices fall
into  $\{0,1,\ldots, N\}$.  Hence, it is necessary to design new
cooperative sensing schemes  for MPTP scenario. In this paper, we
propose two different fusion rules, i.e., the majority fusion and
the optimal fusion.



\subsection{Majority Decision Fusion}

After performing the  local spectrum sensing, SU-$k$ makes its own
decision as ${\cal H}_{i_k}$ and then forwards the index $i_k$ to
the fusion center who combines these results into a $K\times 1$
vector $\mathbf{b}=[j_1,j_2,\ldots, j_k]$. The probability of any specific $\mathbf{b}$
can be easily computed as
\begin{equation}
\text{Pr}(\mathbf{b}|{\cal H}_i)=\prod_{k=1}^K \text{Pr}({\cal
H}_{j_k}^{(k)}|{\cal H}_i),
\end{equation}
 where the superscript $(\cdot)^{(k)}$ represents the
user index. Note that the total number of possible $\mathbf{b}$ is
$(N+1)^K$. From a given $\mathbf{b}$, we can immediately formulate a
voting pool $\vec{d}=(d_{0},\dots,d_{N})$, where $d_{i}$ denotes the
 number of SUs that claim ${\cal H}_i$. Obviously, there
is $\sum_{i=0}^{N} d_{i}=K$. Define the mapping function from
$\mathbf{b}$ to $\vec{d}$ as ${\cal M}(\mathbf{b})=\vec{d}$, which
can be easily obtained in an offline manner once $N$ and $K$ are
fixed. It is not difficult  to find that the total number of
possible $\vec{d}\ $ is ${{K+N}\choose N} = \frac{(K+N)!}{K! N!}$.
The probability of any specific $\vec{d}$ can then be computed as
\begin{align}
\text{Pr}(\vec{d}\ |{\cal H}_i)=\sum_{\mathbf{b}:\ {\cal
M}(\mathbf{b})=\vec{d}} \text{Pr}(\mathbf{b}|{\cal H}_i).
\end{align}

\begin{remark}
If we make the same assumption as did in \cite{ZhangWei}, i.e., the received signal
at each SU experiences almost identical path loss,\footnote{This assumption
 holds  when the distance between any two SUs is small
compared to the distance from PU to any one of the SUs. }  then each
SU has the same decision probability $\text{Pr}({\cal H}_j|{\cal
H}_i), \forall i,j\in\{0,1,\dots,N\}$, and the expression of
$\text{Pr}\big(\vec{d}|{\cal H}_i\big)$ can simplified  as
\begin{align}\label{eq:fusion}
\text{Pr}\big(\vec{d}|{\cal H}_i\big)=&{\!K\!\choose
d_{0}}\text{Pr}({\cal H}_{0}|{\cal H}_i)^{d_{0}}{K\!-\!d_{0}\choose
\!d_{1}\!}\text{Pr}({\cal H}_{1}|{\cal H}_i)^{d_{1}} \cdots
{K-\sum_{l=0}^{N-1}d_{l}\choose d_{N}}\text{Pr}({\cal H}_{N}|{\cal
H}_i)^{d_{N}}\nonumber\\
=&\frac{K!}{\prod_{l=0}^N d_l!}\prod_{n=0}^N\text{Pr}({\cal
H}_{n}|{\cal H}_i)^{d_{n}}.
\end{align}
\end{remark}

A simple and reasonable way to make the decision fusion is to count
the majority claims from SUs, i.e., pick
\begin{equation}
\hat{i}=\arg\max_i\  d_{i}. \label{eq:c1}
\end{equation}
However, a special case happens  when $d_0=\arg\max_{i}\ d_{i}$ while
$d_0<\sum_{i=1}^Nd_i$. In this case, (\ref{eq:c1}) will output
$\hat{i}=0$ and claims the absence of PU, but in fact more users
claim the presence of the PU. Therefore, we should check the
presence of PU before applying the majority rule when the primary
target is to detect the ``on-off'' status of PU.

Let us define $d_\text{off}\triangleq d_0$ and
$d_\text{on}\triangleq\sum_{i=1}^N d_i$. Then, the decision rule can
be expressed as
\begin{equation} \label{eq:c3}
d_\text{on}\mathop{\gtrless}\limits_{\mathcal{H}_\text{off}}^{\mathcal{H}_\text{on}}d_\text{off},
\end{equation}
which can be simplified as
\begin{equation}
d_\text{on}\mathop{\gtrless}\limits_{\mathcal{H}_\text{off}}^{\mathcal{H}_\text{on}}K/2.
\end{equation}
Note that, a special case happens when $K$ is even and
$d_{\text{on}}=K/2$. In this case, the final decision can either be
made as ``on'' or ``off'' because they are equally probable. In this
rest of the discussion, we claim ``on'' if $d_{\text{on}}=K/2$.

If PU is detected to be present, the next step is to discriminate
which power level is in use by majority law
\begin{equation}
\hat{i}=\arg\max_{i\geq1}\  d_{i}. \label{eq:c2}
\end{equation}

In order to fully describe the performance of the majority decision,
we need to refer to the decision probability, denoted as
$\text{Pr}_m({\cal H}_j|{\cal H}_i)$. It can be computed that
\begin{align}
\text{Pr}_m({\cal H}_j|{\cal H}_i)=\sum_{\vec{d}\in{\cal S}_{m_j}}
\text{Pr}(\vec{d}\ |{\cal H}_i),\label{eq:Prm}
\end{align}
where the set ${\cal S}_{m_j}$ is defined as
\begin{align}
{\cal S}_{m_j}=& \left\{ \begin{array}{ll} \{\vec{d}\ |\ d_0>K/2\}
&\mbox{if $j=0$}
\\\{\vec{d}\ |\ d_j=\max\{d_1,\dots,d_N\},\ d_0\leq K/2\}&\mbox{if $j\geq1$}
\end{array} \right..
\end{align}
There exist special cases when more than one state ${\hat i}$
simultaneously achieve the maximum number of votes. In this
situation, one can  choose any of them as the final decision since
they are equally probable. In this paper, we always choose the
largest value of $\hat{j}$  as the final decision if this happens,
and all the theoretical and numerical results in the rest of this
paper are based on this consideration.

The set  ${\cal S}_{m_j}$ can be obtained from a mapping function in
an offline manner. Hence, we can easily build tables and mapping
functions as illustrated in Fig. \ref{fig:mapping}, and compute
$\text{Pr}_m({\cal H}_j|{\cal H}_i)$ \emph{a prior}. Fortunately, a
more explicit expression for  $\text{Pr}_m({\cal H}_j|{\cal H}_i)$
in majority voting can be derived as in the following theorem.



\begin{theorem}\label{close-form}
The majority decision fusion has the decision probability,
\begin{align}\label{eq:maj}
\nonumber &\text{Pr}_m({\cal H}_j|{\cal
H}_i)\\
&= \sum_{d_0=0}^{\lfloor\frac{K}{2}
\rfloor}\sum_{d_j=\lceil\frac{K-d_0}{N}\rceil}^{K-d_0}
\sum_{d_1=max\{0,\alpha_1\}}^{min\{d_j,\beta_1\}}\cdot\cdot\sum_{d_{j-1}=max\{0,\alpha_{j-1}\}}^{min\{d_j,\beta_{j-1}\}}
\sum_{d_{j+1}=max\{0,\alpha_{j+1}\}}^{min\{d_j-1,\beta_{j+1}\}}\cdot\cdot\sum_{d_N=max\{0,\alpha_N\}}^{min\{d_j-1,\beta_N\}}
\!\!\text{Pr}\big(\vec{d}|\mathcal{H}_i\big)
\end{align}
for $j\geq1$, and
\begin{align}
\text{Pr}_m({\cal H}_0|{\cal H}_i)=\sum_{d_0=\lfloor \frac{K}{2}
\rfloor+1}^{K}\sum_{d_1=0}^{K-d_0}\sum_{d_2=0}^{K-d_0-d_1}\ldots
\sum_{d_{N-1}=0}^{K-\sum_{l=0}^{N-2}d_l}
\text{Pr}\big(\vec{d}|\mathcal{H}_i\big),\label{eq:maj2}
\end{align}
where $\lceil\cdot\rceil$, and $\lfloor\cdot\rfloor$ denote the ceiling function and floor function, respectively. Moreover,
$\alpha_n$ and $\beta_n$ are defined as
\begin{equation}\label{alpha beta define}
\begin{aligned}
\alpha_n =& \left\{ \begin{array}{ll}
K-\sum_{i=0}^{n-1}d_i-(N-n)d_j+N-j &\mbox{if $1\leq n<j$}
\\K-\sum_{i=0}^{n-1}d_i-(N-n)d_j+N-n &\mbox{if $j<n\leq N$}
\end{array} \right.,
\\\beta_n =& \left\{ \begin{array}{ll} K-\sum_{i=0}^{n-1}d_i-d_j &\mbox{if $1\leq n<j$}
\\K-\sum_{i=0}^{n-1}d_i &\mbox{if $j<n\leq N$}
\end{array} \right..
\end{aligned}
\end{equation}
\end{theorem}

\begin{proof}
To calculate the decision probability from (\ref{eq:Prm}),
we need to find all candidates $\vec{d}$ in ${\cal S}_{m_j}$. In other words, we need to determine
the range of elements  $d_j$'s, $j=0,1,\ldots,N$ in $\vec{d}$.

Let us start form  $j\geq1$. As shown in
$\mathcal{S}_{m_j}$, $\vec{d}$ must satisfy $d_0\leq K/2$, so
the range of $d_0$ should be from $0$ to
$\lfloor\frac{K}{2}\rfloor$. Moreover, since
$d_{j}=\max\{d_1,\dots,d_N\}$, the lower bound of $d_j$ must be no less than $\frac{K-d_0}{N}$, otherwise
there will always be another $d_{j_1}\neq d_j$ but satisfies $d_{j_1}>d_j$. Therefore the range of $d_j$ is from
$\lceil\frac{K-d_0}{N}\rceil$ to $K-d_0$.
We then separately determine the range of $d_n$ for $1\leq n<j$ and $j<n\leq N$, respectively.

1) When $n<j$, the upper bound of $d_n$ must be less than or equal
to the unassigned value of $K$, which is $K$ minus all the values
that have already been assigned to $d_i, 0\leq i<n$ and $d_j$, i.e.
$K-\sum_{i=0}^{n-1}d_i-d_j$. Bearing in mind the the constraint
$d_n\leq d_j$, the upper bound can be expressed as
\begin{align*}
d_{n}^{\text{upper}}=\min\left\{d_j,\
K-\sum_{i=0}^{n-1}d_i-d_j\right\},\quad n<j.
\end{align*}
As for the lower bound, $d_n$ should not be too small to allow any
other undetermined $d_k (n<k\leq N$ and $k\neq j)$ be greater than
$d_j$. The extreme case happens when  all the undetermined $d_k$'s get
their highest values, i.e., $d_k$ equals to $d_j$ for $n<k<j$ while $d_k$ is $d_j-1$ for
$j<k\leq N$. Then the summation of these $d_k$ is
$(j-n-1)d_j+(N-j)(d_j-1)$. Combining this result  as well as the
constraint $d_n\geq0$, we get the lower bound of $d_n$ as
\begin{align*}
d_{n}^{\text{lower}}=\max\left\{K-d_j-\sum_{i=0}^{n-1}d_i-(j-n-1)d_j
-(N-j)(d_j-1),\ 0 \right\},\quad n<j.
\end{align*}

2) When $n>j$, the maximum value of $d_n$ can only be $d_j-1$, thus
the upper bound of the summand changes to
\begin{align*}
d_{n}^{\text{upper}}=\min\left\{d_j-1,\ K-\sum_{i=0}^{n-1}d_i\right\},\quad n>j.
\end{align*}
Similar to the previous discussion, the maximum summation of all the
undetermined $d_k(n<k\leq N)$ in this situation is $(N-n)(d_j-1)$,
so the lower bound is
\begin{align*}
d_{n}^{\text{lower}}=\max\left\{K-\sum_{i=0}^{n-1}d_i-(N-n)(d_j-1),\
0\right\},\quad n>j.
\end{align*}

When we use $\alpha_n$ and $\beta_n$ in (\ref{alpha beta define}) to
simplify the expression of the range of $d_n$'s, the   equality (\ref{eq:maj}) for
$j \geq 1$ is proved.

As for $j=0$, the only constraint
for $\vec{d}$ is $d_0>K/2$ from $\mathcal{S}_{m_j}$, so the
summation range of $d_0$ must be from $\lfloor \frac{K}{2}\rfloor+1$
to $N$. All the others  $d_n,
n\in\{1,2,\ldots,N\}$ can be freely chosen as long as $\sum_{i=0}^N d_i=K$. If we assign the
values for $d_i$'s one by one, then for any $d_n$, its lowest possible value is 0 while its highest
possible value is $K-\sum_{i=0}^{n-1} d_i$. Note that, $d_N$ is a fixed value
when all the previous $d_i, 0\leq i<N$ are chosen and does not need to be included in
the summand.  Then, the   equality (\ref{eq:maj2})
for $j=0$ is proved.
\end{proof}

Once $\text{Pr}_m({\cal H}_j|{\cal H}_i)$ is derived, then the false
alarm, the detection probability as well as the discrimination
probability for majority cooperation can be immediately  obtained
as the summation of corresponding $\text{Pr}_m({\cal H}_j|{\cal
H}_i)$.

\begin{remark} Though majority
decision fusion rule has been widely accepted in many research
areas, the analytical approach to study its performance, e.g, obtaining
$\text{Pr}_m({\cal H}_j|{\cal H}_i)$ from $\text{Pr}({\cal H}_j|{\cal H}_i)$,  has never been fully discussed to the best of
the authors' knowledge.
\end{remark}

%




In Fig. \ref{fig:decision_prob}, we provide one example to verify
several $\text{Pr}_m({\cal H}_j|{\cal H}_i)$ with $N=4$, $K=5$ and
average SNR=$-12$dB. It is  clearly seen that the numerical results
match the theoretical ones very well.

\subsection{Optimal Decision Fusion}

Though the majority decision fusion rule is
 very simple and  effective, it
does have some drawbacks which limit the performance. For example,
when $\text{Pr}({\cal H}_i)\gg\text{Pr}({\cal H}_j)$, even if the
detection result is $d_i<d_j$,  it is still  possible that ${\cal
H}_i$ is truer than ${\cal H}_j$. The reason  is that majority
decision is a type of `hard'' decision and is not ``soft'' enough to
count in the prior probability of each hypothesis.

From the probabilistic point of view, the optimal decision fusion with the observation $\vec{d}$ should follow MAP criterion \cite{Kay}. Similar to
majority decision fusion, we need to first to make a decision about
the presence of PU before recognizing the power levels, i.e.,
\begin{equation}\label{eq:opt fusion}
\text{Pr}\big(\mathcal{H}_\text{off}|\vec{d}\
\big)\mathop{\gtrless}\limits_{\mathcal{H}_\text{on}}^{\mathcal{H}_\text{off}}\text{Pr}\big(\mathcal{H}_\text{on}|\vec{d}\
\big).
\end{equation}
From Bayes rule, there is
\begin{align}
\text{Pr}\big({\cal H}_n|\vec{d}\ \big)=\frac{\text{Pr}\big(\vec{d}\
|{\cal H}_n\big)\text{Pr}({\cal H}_n)}{\text{Pr}\big(\vec{d}\
\big)}.
\end{align}
Hence, (\ref{eq:opt fusion}) can be simplified as
\begin{equation}
\text{Pr}\big(\vec{d}\
|\mathcal{H}_0\big)\text{Pr}(\mathcal{H}_0)\mathop{\gtrless}\limits_{\mathcal{H}_\text{on}}^{\mathcal{H}_\text{off}}\sum_{i=1}^N
\text{Pr}\big(\vec{d}\ |\mathcal{H}_i\big)\text{Pr}(\mathcal{H}_i).
\end{equation}

If PU is detected to be present, we continue to recognize the power
level of PU. Following the similar steps from (\ref{MAP multiple})
to (\ref{MAP multiple_simplified}), the detection rule is
\begin{equation}
\hat i=\arg\max_{i\geq1}\ \text{Pr}\big({\cal H}_i|\vec{d}\
\big)=\arg\max_{i\geq1}\ \text{Pr}\big(\vec{d}\ |{\cal
H}_i\big)\text{Pr}({\cal H}_i). \label{eq:opti}
\end{equation}

The decision probability of the optimal decision fusion can be
expressed as
\begin{align}
&\text{Pr}_o({\cal H}_j|{\cal H}_i)=\sum_{\vec{d}\in{\cal
S}_o}\text{Pr}(\vec{d}\ |{\cal H}_i)
\end{align}
where the set ${\cal S}_{o_j}$ is defined as
\begin{align*}
{\cal S}_{o_j}=&\left\{ \begin{array}{ll} \{\vec{d}\ \big|\
\text{those}\ \vec{d}\ \text{that claim absence from}\ (\ref{eq:opt
fusion})\} &\mbox{if $j=0$}
\\ \{\vec{d}\ \big|\ \text{those}\ \vec{d}\ \text{that result in}\ \hat{i}=j\ \text{in}\ (\ref{eq:opti})\ \text{but claim presence from}\ (\ref{eq:opt fusion})\}  &\mbox{if $j\geq1$}
\end{array} \right.
\end{align*}
The elements in  ${\cal S}_{o_j}$ is an implicit function of
$\text{Pr}({\cal H}_i)$ and $\text{Pr}({\cal H}_j|{\cal H}_i)$,
which makes it difficult to obtain an explicit expression of
$\text{Pr}_o({\cal H}_j|{\cal H}_i)$. Nevertheless,
$\text{Pr}_o({\cal H}_j|{\cal H}_i)$ is  the summation of those
$\text{Pr}(\vec{d}\ |{\cal H}_i)$ whose $\vec{d}$ could result in
the decision of $j$ and these $\vec{d}$ can be found from a
predetermined mapping,
as did in Fig. \ref{fig:mapping}. 

\begin{remark}
Since $\text{Pr}({\cal H}_i)$ and $\text{Pr}({\cal H}_j|{\cal
H}_i)$ are real continuous values, the probability for obtaining
more than one maximum index  from (\ref{eq:opti}) is 0, and the
corresponding discussion is not necessary.
\end{remark}

\begin{remark}
If we assume the same fading gain for all SUs \cite{ZhangWei} and
apply (\ref{eq:fusion}), then a more concise form of the decision (\ref{eq:opti})
can be obtained as
\begin{align}
\hat{i}=&\arg\max_{i\geq1}\ \text{Pr}({\cal
H}_i)\prod_{n=0}^{N}\text{Pr}({\cal H}_n|{\cal
H}_i)^{d_{n}}=\arg\max_{i\geq1}\ \log\text{Pr}({\cal
H}_i)+\sum_{n=0}^{N}d_n\log\text{Pr}({\cal H}_n|{\cal
H}_i).
\end{align}
\end{remark}

\section{Simulations}\label{sec:simulation}

In this section, we resort to numerical examples  to evaluate the
proposed studies. Four levels of primary transmit power are assumed,
while the corresponding prior probabilities are set as
$\text{Pr}({\cal H}_0)=0.5$, and $\text{Pr}({\cal H}_i)=0.125,
i=1,2,3,4$. The channel gain $\gamma$ and the noise variances are
taken as units. The power levels satisfy $P_1:P_2:P_3:P_4=3:5:7:9$,
and the average SNR is defined as
 $\frac{1}{4}\sum_{i=1}^4 P_i/{\sigma^2_n}$.

\subsection{Sensing with A Single SU}

In Fig. \ref{fig:detection_prob},  we evaluate the performance of
detecting the presence of PU versus the number of   samples for the
proposed sensing strategies. It is seen that sensing strategy-I
works better than sensing strategy-II, especially when the sampling
number is small, which matches our discussion in Section
\ref{sec:rationale} that $\theta_\text{on/off}<\phi_\text{on/off}$.
Nevertheless, the gaps between the two sensing strategies reduce
when the number of samplings become larger or the PU's SNR becomes
higher. This implies that when the sensing conditions becomes
better, then difference of the two sensing strategies gradually
diminishes and the choice of sensing strategies becomes less
important.

One the other hand,  Fig. \ref{fig:discrimination_prob} displays the
performance of discriminating the power level versus the number of
samples of the proposed  sensing strategies. The discrimination
probability follows the definition in (\ref{eq:P_dis2}), i.e., we
treat absence as an equivalent power level with 0 value. From
Fig.~\ref{fig:discrimination_prob}, we see that sensing strategy-II
works slight better than strategy-I but the difference diminishes
when SNR becomes higher. This phenomenon right fits our analysis in
subsection \ref{sec:rationale} that strategy-II takes into account
all error cases.

Next we demonstrate the sensing performance versus SNR in Fig.
\ref{fig:singleSU_vs_SNR} for both the detection probability and the
discrimination probability defined in (\ref{eq:P_dis1}). It is seen
that sensing strategy-I out performs sensing strategy-II in terms of
both detection probability and discrimination probability. For
sensing strategy-I, the difference between detection probability and
discrimination probability is very large at low SNR. The reason is
that even if PU is detected to be present, the strategy-I actually
makes many mistakes about PU's actual power level. However, for
sensing strategy-II, $\mathcal{H}_0$ may mask all the other states
when SNR is low, which almost ruins the detecting ability.
Nevertheless, since sensing strategy II is originally defined for
discriminating all $\mathcal{H}_i,i\geq 0$, once we include
$\mathcal{H}_0$, the discrimination probability could outperform
sensing strategy-I as has been demonstrated in Fig.
\ref{fig:discrimination_prob}.

\subsection{Cooperative Sensing}

For cooperative sensing examples, the local sensing results of all
users are made from strategy-I.

In Fig. \ref{fig:detect_and_discrim_Prob}, we show the detection
probability as well as the discrimination probability versus the
number of receive samples when five SUs cooperative to make the
final decision.  The SNR is taken as -12 dB. Compared with the
sensing performance of single local SU in Fig.
\ref{fig:detection_prob} and Fig. \ref{fig:discrimination_prob}, it
is clearly seen that the performance is greatly improved when
cooperative scheme is applied. Moreover,  the optimal decision
outperforms the majority decision at all sample numbers.
Nevertheless, the optimal decision needs to dynamically build the
mapping function and is not as simple as the majority decision
fusion.

The average error detection probability of the two fusion rules
versus the number of samples is displayed in Fig.
\ref{fig:Error_Detection_Prob}. Five SUs cooperate to make the final
decision and SNR is taken as -12 dB. Moreover, $\delta$ denotes the
offset of index between the true power level and the detected power
level, e.g., the detection error probability with $\delta=1$ is the
summation of those $\text{Pr}({\cal H}_{j}|{\cal
H}_i)\text{Pr}(\mathcal{H}_i)$ satisfying $|i-j|=1$, i.e.,
$\sum_{|i-j|=1}\text{Pr}({\cal H}_{j}|{\cal
H}_i)\text{Pr}(\mathcal{H}_i)$.  It is seen that the error
probability decreases extensively when $\delta$ increases. This is
not unexpected because the chances for making  a wrong decision to
the farer power level should be smaller. A very important indication
to practical design is that,  SUs may only pay
 attention to those errors with smaller $\delta$ and set the corresponding cost
 values in  Bayes Risk detection.

Next we show the performance of cooperative sensing versus SNR in
Fig. \ref{fig:CoSensing_Pd_Pdis_vs_SNR}. Not surprisingly, the
optimal decision fusion rule outperforms the majority decision
fusion in all SNR range. Besides, the discrimination probability for
both sensing rules also get closer to the detection probability as
the SNR grows. Moreover, the gaps between the optimal fusion and the
majority fusion diminish when SNR becomes larger.

In the last example, we show the detection probability and
discrimination probability versus the number of SU in Fig.
\ref{fig:4l_Pd_and_Pdis_vs_K}, with SNR$=-12$dB and $M=5000$.
Obviously, both decision rules provide better results when the
number of the SUs increases and the optimal decision rule always
outperforms the majority rule, which matches our intuition very
well. Moreover, increasing the number of SU after a certain amount
may not be very helpful, e.g., eight SUs in majority sensing
strategy.

\section{Conclusions}

In this paper, we investigated a new CR scenario, i.e., MPTP, that
embraces multiple primary powers, which both matches the practical
transmission and fits the theocratical demands of adapting the
transmit power. We designed two different spectrum sensing
strategies  which are shown to possess different but correlated
optimization criteria. Most results, e.g., threshold expressions,
probabilities, are derived in closed-forms. We present a thorough
discussion over all kinds of aspects of the new spectrum sensing
strategies, including the power mask effects and its reasoning, the
new definition of performance metrics, as well as the rationales
behind.  Moreover, we developed two different cooperative sensing
algorithms which are shown to be very different from the traditional
cooperative schemes. Various simulations are provided lately to
corroborate the proposed studies. It is then believed that there
could exist many new problems in MPTP waiting for exploitation,
while in the mean time, many existing studies for traditional CR
deserve re-investigation.

\newpage

\begin{table}[t]
\centering \caption{\label{Tab:1}Protected Service Contour Levels}
\begin{tabular}{|c|c|c|}
\hline
Type of TV & Band/Channel & Protected \\
station &  & Contour E-field \\
& & Level(dBu) \\
\hline
Full Power & Low VHF (2-6) & 47\\
\cline{2-3}
Analog TV & High VHF (7-13) & 56 \\
\cline{2-3}
& UHF (14-69) & 64 \\
\hline
Low Power & Low VHF (2-6) & 62\\
\cline{2-3}
Analog TV & High VHF (7-13) & 68 \\
\cline{2-3}
& UHF (14-69) & 74 \\
\hline
Full Power & Low VHF (2-6) & 28\\
\cline{2-3}
Digital TV & High VHF (7-13) & 36 \\
\cline{2-3}
& UHF (14-51) & 41 \\
\hline
Low Power & Low VHF (2-6) & 43\\
\cline{2-3}
Digital TV & High VHF (7-13) & 48 \\
\cline{2-3}
& UHF (14-51) & 51 \\
\hline
\end{tabular}
\end{table}


\begin{figure}[t]
\centering
\includegraphics[width=100mm]{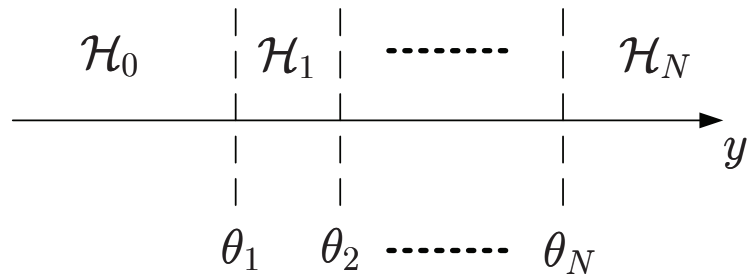}
\caption{Multiple power level detections from multiple thresholds.} 
\label{fig:thresholds}
\end{figure}

\begin{figure}[t]
\centering
\includegraphics[width=150mm]{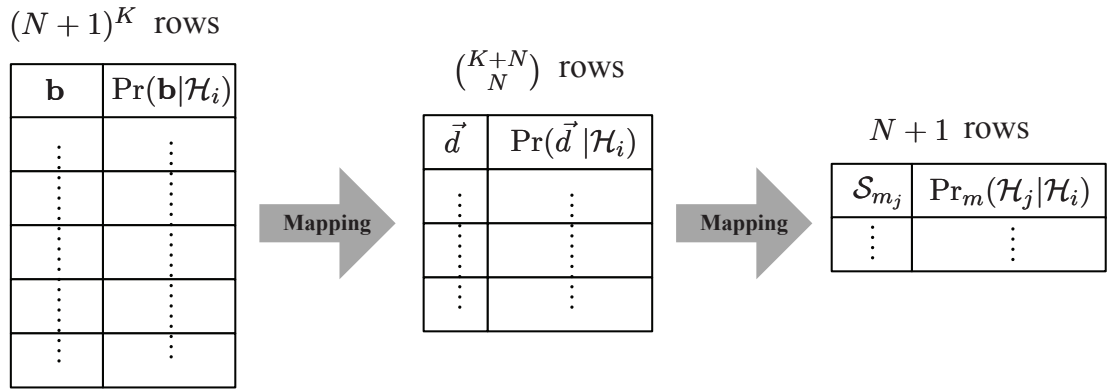}
\caption{Illustration of mapping from $\mathbf{b}$ to $\vec{d}$ and then to ${\cal S}_{m_j}$.} 
\label{fig:mapping}
\end{figure}

\begin{figure}
\centering
\includegraphics[width=120mm]{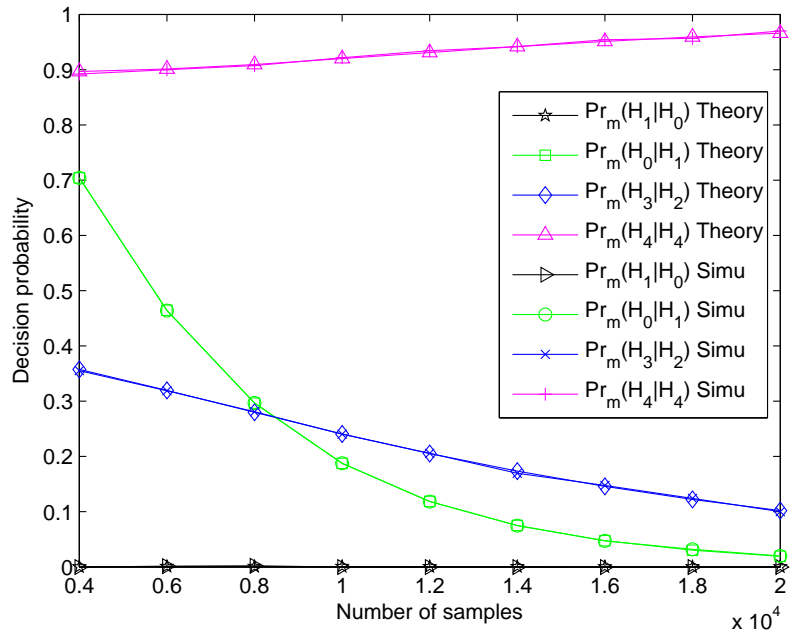}
\caption{Theoretical analysis and numerical results for decision probability under majority decision fusion.} 
\label{fig:decision_prob}
\end{figure}

\begin{figure}
\centering
\includegraphics[width=120mm]{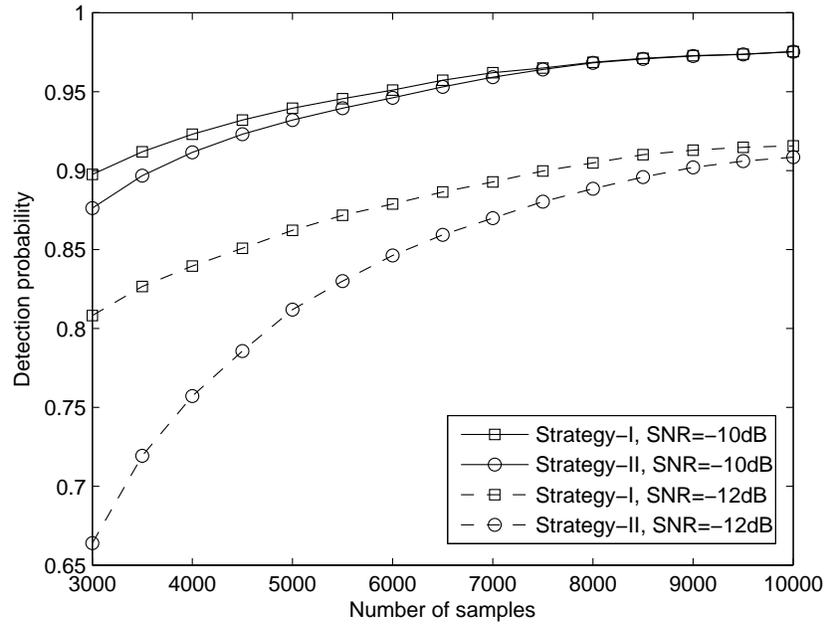}
\caption{The detection probability of local SU versus the number of samples with average SNR$=-10\text{dB},-12\text{dB}$.} 
\label{fig:detection_prob}
\end{figure}

\begin{figure}
\centering
\includegraphics[width=120mm]{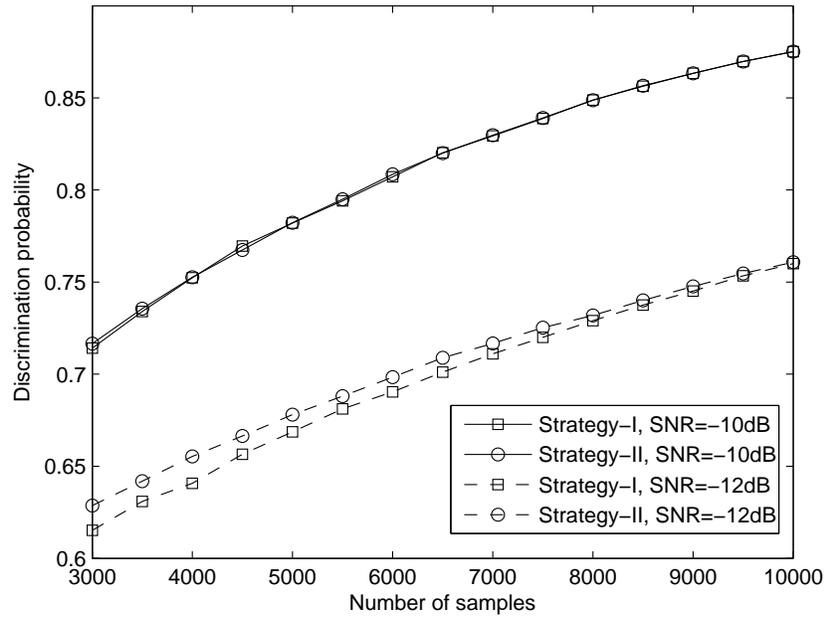}
\caption{The discrimination probability of local SU versus the number of samples.} 
\label{fig:discrimination_prob}
\end{figure}

\begin{figure}
\centering
\includegraphics[width=120mm]{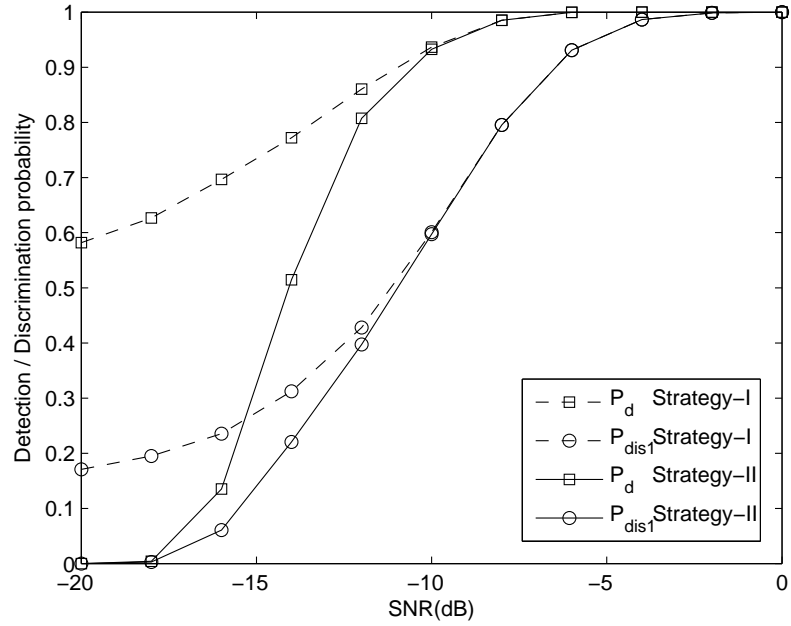}
\caption{The detection and discrimination probability of local SU versus SNR with M=5000.} 
\label{fig:singleSU_vs_SNR}
\end{figure}

\begin{figure}
\centering
\includegraphics[width=120mm]{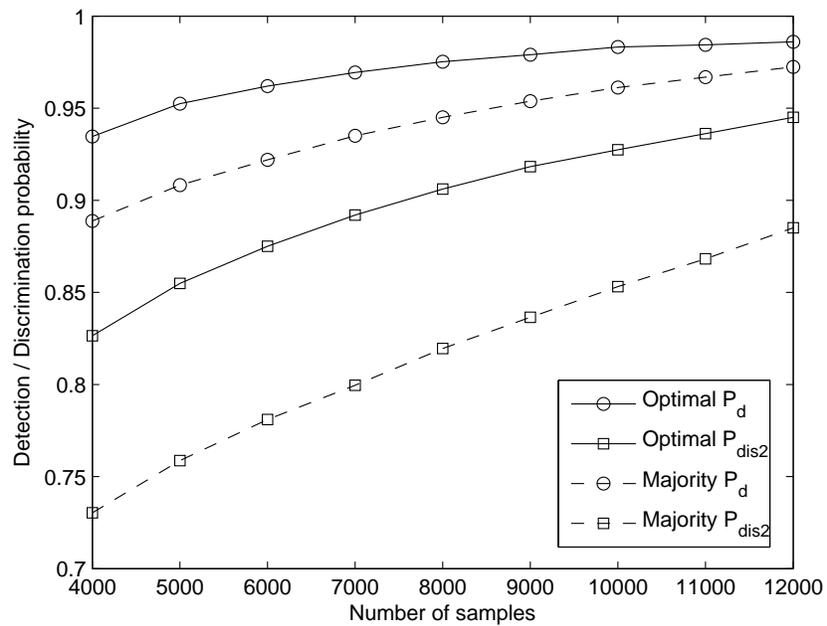}
\caption{Detection probability and discrimination probability  versus number of samples for these two cooperative sensing methods with $K=5$, SNR$=-12$dB.} 
\label{fig:detect_and_discrim_Prob}
\end{figure}

\begin{figure}
\centering
\includegraphics[width=120mm]{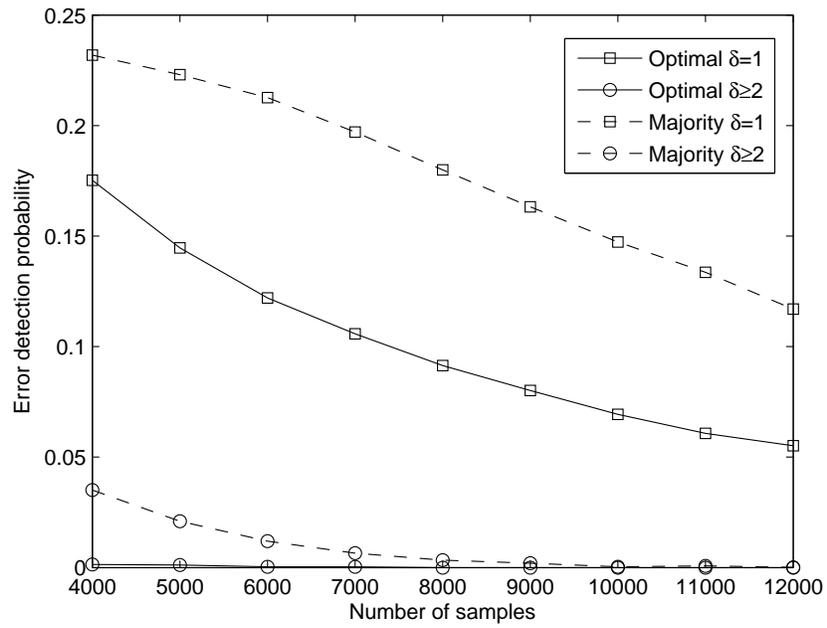}
\caption{Error detection probability versus the number of samples for $K=5$, SNR$=-12$dB.} 
\label{fig:Error_Detection_Prob}
\end{figure}

\begin{figure}
\centering
\includegraphics[width=120mm]{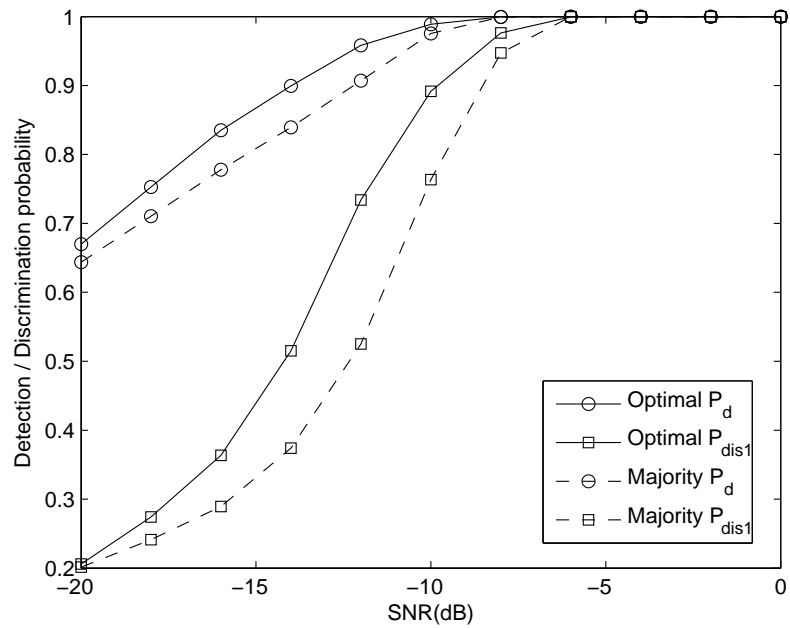}
\caption{Detection probability and discrimination probability in
cooperative sensing versus SNR with $K=5$, $M=5000$.}
\label{fig:CoSensing_Pd_Pdis_vs_SNR}
\end{figure}

\begin{figure}
\centering
\includegraphics[width=120mm]{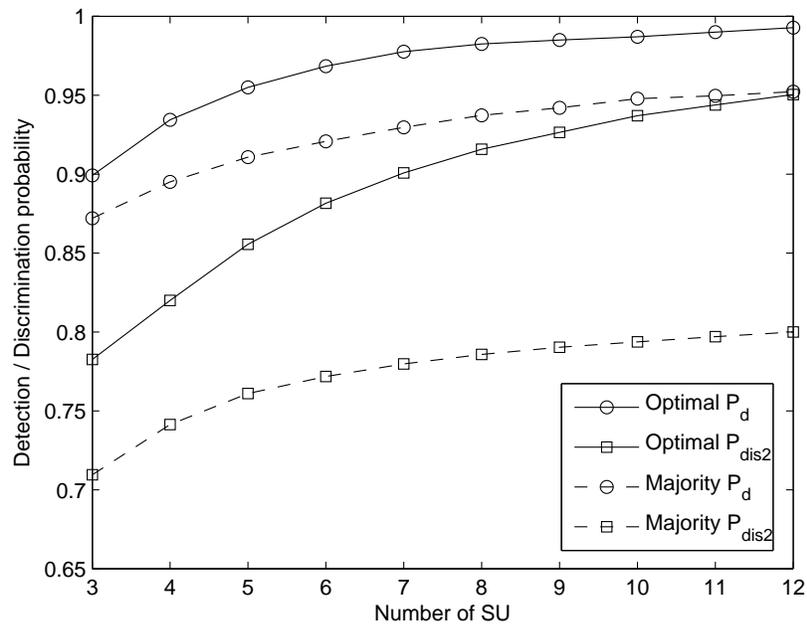}
\caption{Detection probability and discrimination probability versus
the number of SU for SNR$=-12$dB, $M=5000$.}
\label{fig:4l_Pd_and_Pdis_vs_K}
\end{figure}

\end{document}